\documentclass[conference]{IEEEtran}
\usepackage{cite}
\usepackage{amsmath,amssymb,amsfonts, float}
\usepackage[linesnumbered,ruled,vlined]{algorithm2e}
\usepackage{bbm}
\usepackage{graphicx}
\usepackage{textcomp}
\usepackage{nicefrac}
\usepackage{xcolor}
\usepackage[utf8]{inputenc}
\usepackage[english]{babel}
\usepackage{support-caption}
\usepackage{amsthm, multirow, subcaption}
\usepackage{comment}
\usepackage{mathtools}
\allowdisplaybreaks
\newtheorem{assumption}{Assumption}
\newcommand{\instance}{(M,N,\boldsymbol \mu)}

\newcommand{\PP}{\mathbb{P}}
\newcommand{\E}{\mathbb{E}}

\newcommand{\nn}{\nonumber}

\newcommand{\calP}{\mathcal{P}}

\newcommand{\OO}{\text{O}}

\usepackage{ifthen}
\usepackage{color}

\newtheorem{conjecture}{Conjecture}

\newtheorem{lemma}{Lemma}
\newtheorem{remark}{Remark}

\newtheorem{theorem}{Theorem}
\newtheorem{definition}{Definition}

\def\BibTeX{{\rm B\kern-.05em{\sc i\kern-.025em b}\kern-.08em
    T\kern-.1667em\lower.7ex\hbox{E}\kern-.125emX}}
\begin{document}
\title{Decentralized Age-of-Information Bandits}

\author{\IEEEauthorblockN{Archiki Prasad}
\IEEEauthorblockA{\textit{Department of Electrical Engineering  } \\
\textit{Indian Institute of Technology Bombay}\\
archiki@iitb.ac.in}
\and
\IEEEauthorblockN{Vishal Jain}
\IEEEauthorblockA{\textit{Department of Electrical Engineering  } \\
\textit{Indian Institute of Technology Bombay}\\
vishalj409@gmail.com}
\and
\IEEEauthorblockN{Sharayu Moharir}
\IEEEauthorblockA{\textit{Department of Electrical Engineering  } \\
\textit{Indian Institute of Technology Bombay}\\
sharayum@ee.iitb.ac.in}
}

\maketitle
\thispagestyle{plain}
\pagestyle{plain}

\begin{abstract}
Age-of-Information (AoI) is a performance metric for scheduling systems that measures the freshness of the data available at the intended destination. AoI is formally defined as the time elapsed since the destination received the recent most update from the source. We consider the problem of scheduling to minimize the cumulative AoI in a multi-source multi-channel setting. Our focus is on the setting where channel statistics are unknown and we model the problem as a distributed multi-armed bandit problem. For an appropriately defined AoI regret metric, we provide analytical performance guarantees of an existing UCB-based policy for the distributed multi-armed bandit problem. In addition, we propose a novel policy based on Thomson Sampling and a hybrid policy that tries to balance the trade-off between the aforementioned policies. Further, we develop AoI-aware variants of these policies in which each source takes its current AoI into account while making decisions. We compare the performance of various policies via simulations.
\end{abstract}


\section{Introduction}
We consider the problem of scheduling in a multi-source multi-channel system, focusing on the metric of Age-of-Information (AoI), introduced in \cite{kaul2011minimizing}. AoI is formally defined as the time elapsed since the destination received the recent most update from the source. It follows that AoI is a measure of the freshness of the data available at the intended destination which makes it a suitable metric for time-sensitive systems like smart homes, smart cars, and other IoT based systems. Since its introduction, AoI has been used in areas like caching, scheduling, and channel state information estimation. A comprehensive survey of AoI-based works is available in \cite{kosta2017age}. 

The work in \cite{bhandari2020age} shows the performance of AoI bandits for a single source and multiple channels, where the source acts as the ``bandit" which pulls one of the arms in every time-slot, i.e., selects one of the channels for communication. The aim is to find the best arm (channel) while minimizing the AoI (instead of the usual reward maximization) for that source. In this work, we address more practical scenarios in which systems have multiple sources looking to simultaneously communicate through a common, limited pool of channels. Here, we need to ensure that the total sum of the AoIs of all the sources is minimized. Moreover, we consider a decentralized setting where the sources cannot share information with each other, meaning that chances of collisions (attempting to communicate through the same channel in a time-slot) can be very high. Thus, the task of designing policies that minimize total AoI while ensuring fairness among all sources and avoiding collisions is a challenging problem. Prior works on decentralized multi-player MAB problems primarily discuss reward-based policies, however, minimizing AoI is more challenging, as the impact of sub-optimal decisions gets accumulated over time.

The decentralized system comprises of multiple sources and multiple channels, where at every time slot, each of the $M$ decentralized users searches for idle channels to send a periodic update. The probability of an attempted update succeeding is independent across communication channels and independent and identically distributed (i.i.d) across time-slots for each channel. AoI increases by one on each failed update and resets to one on each successful update. These distributed players can only learn from their local observations and collide (with a reward penalty) when choosing the same arm. The desired objective is to develop a sequential policy running at each user to select one of the channels, without any information exchange, in order to minimize the cumulative AoI over a finite time-interval of $T$ consecutive time-slots.

\subsection{Our Contribution}
\emph{Optimality of Round-Robin policy:} We describe an oracle Round-Robin policy and characterize its optimality.

\emph{Upper-Bound on AoI regret of DLF~\cite{OpportunisticSpectrum}:} We characterize a generic expression for the upper bound of the total cumulative AoI for any policy. Further, we show that the AoI regret of the DLF policy scales as $\OO(M^2 N \log^2 T)$.

\emph{New AoI-agnostic policies:} We propose a Thompson Sampling~\cite{thompson1933likelihood} based policy. We also present a new hybrid policy trading-off between Thompson Sampling and DLF.

\emph{New AoI-aware policies:} We propose AoI-aware variants for all the AoI-agnostic policies. When AoI values are below a certain threshold, the variants mimic the original policy. Otherwise, the variants exploit local past observations. Through simulations, we show that these variants exhibit lower AoI regrets as compared to their agnostic counterparts.

\subsection{Related Work}
In this section, we discuss the prior work most relevant to our setting. AoI-based scheduling has been explored by \cite{sombabu2018age, tripathi2017age, tripathi2019whittle, jhunjhunwala2018age, kadota2018optimizing}, where the channel statistics were assumed to be known and an infinite time-horizon steady-state performance-based approach is adopted. \cite{bhandari2020age} explores the setting where the channel statistics are unknown, for a single source and multiple channels. We consider a decentralized system with multiple sources, which is a much more complex setting.

Time Division Fair Sharing algorithm proposed in \cite{liu2010distributed} addresses the fair-access problem in a distributed multi-player setting. This policy was outperformed by the policy proposed in \cite{OpportunisticSpectrum}. Recent work in Multi-Player MABs (MPMABs) include \cite{besson2018multi, hanawal2018multi, boursier2019sic}. The policies in \cite{besson2018multi} consider an alternate channel allotment in the event of a collision. Although these ideas can be adopted in all our policies, in this work, we did not wish to make any assumptions about the feasibility of these alternate allotments. The settings in \cite{hanawal2018multi, boursier2019sic} are significantly different and cannot be readily adapted to our setting. Further, the policies may not be directly consistent with the AoI metric.
\begin{figure}[t!]
\centering
\includegraphics[trim=65 0 0 0,clip,scale=0.6]{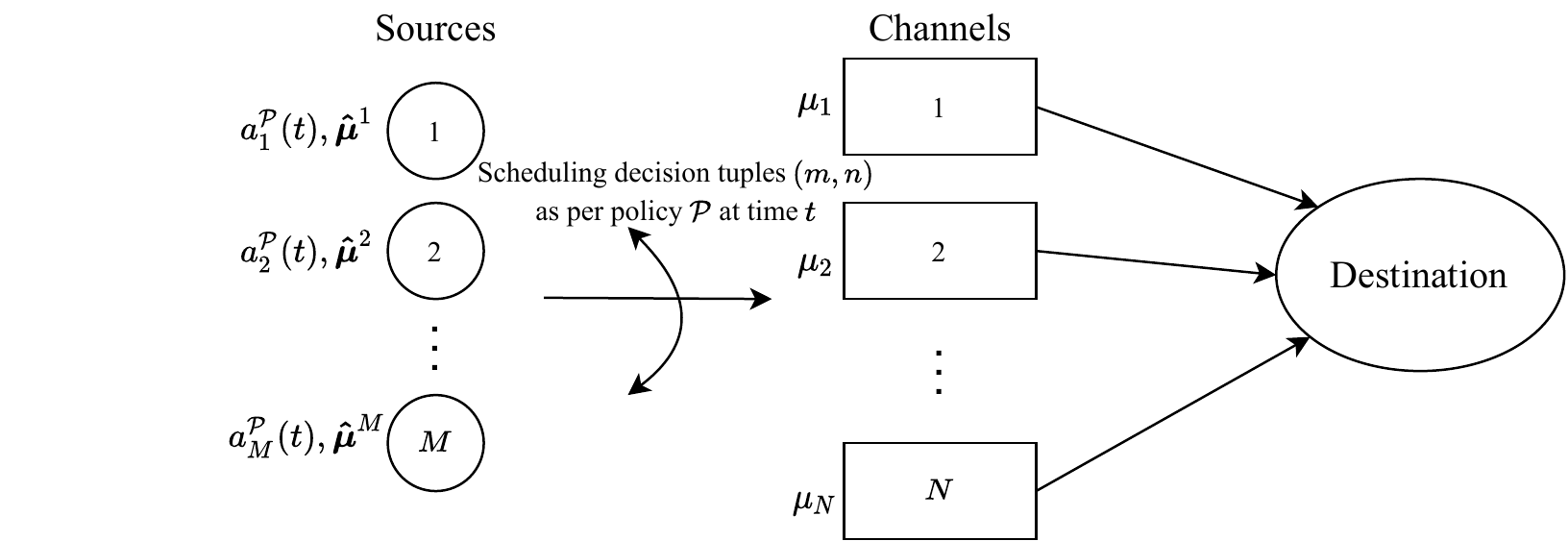}
\caption{Schematic of our system scheduling $M$ sources on $N$ channels based on local channel estimates at time-step $t$}
\label{fig:system}
\end{figure}
\section{Setting}
\subsection{Our System}
We consider a system with $M$ sources and $N$ channels ($N \geq M$). The sources track/measure different time-varying quantities and relay their measurements to a monitoring station via available communication channels. Time is divided into slots. In each time slot, each of the $M$ decentralized users (sources) selects an arm (channel) only based on its own observation histories under a decentralized policy, and then attempts to transmit via that channel.\footnote{We use the terms source and user, and arm and channel interchangeably} Each attempted communication via channel $i$ is successful with probability $\mu_i$ and unsuccessful otherwise, independent of all other channels and across time-slots. The values of $\mu_i$s are unknown to the sources. We use $\boldsymbol \mu$ to denote the channel transmission success probabilities (also denoted as $[ \mu_1; \mu_N ]$). Without loss of generality, we assume $\boldsymbol \mu$ to be in descending order, i.e., $\mu_1 > \mu_2 > ... > \mu_N$.  Figure~\ref{fig:system} shows a  schematic of our system in function.

A decentralized setting means that when a particular arm $i$ is selected by user $j$, the reward (whether the transmission was successful or not) is only observed by user $j$, and if no other user is playing the same arm, a reward is obtained with probability $\mu_j$. Else, if multiple users are playing the same arm (which is possible in a decentralized setting), then we assume that, due to collision, exactly one of the conflicting users can use the channel to transmit, while the other users get zero reward. This ``winner" is chosen at random since we assume all sources to have equal priority. This is consistent with network protocols like CSMA~\cite{kurose2010computer} with perfect sensing.

A \emph{scheduling algorithm} is run at time $t$ for a source $m$ independently by using the local channel estimate $\boldsymbol{\hat\mu}^m$ and the number of times a particular channel $n$ was utilized so far. The output of this algorithm is a channel on which the source $m$ can be scheduled (if no collisions). Our scheduling algorithms (referred to as policies) are discussed in Sections \ref{sec:analytic} and \ref{sec:mainResults}.

\subsection{Metric: Age-of-information Regret}

The age-of-information is a metric that measures the freshness of information available at the monitoring station. It is formally defined as follows. 
\begin{definition}[Age-of-Information (AoI)]
  Let $a(t)$ denote the AoI at the monitoring station in time-slot $t$ and $u(t)$ denote the index of the time-slot in which the monitoring station received the latest update from the source before the beginning of time-slot $t$. Then, 
$$a(t) = t- u(t).$$  
\end{definition}

By definition, 
\begin{align*}
a(t) = 
\begin{cases}
1 & \text{if the update in slot $t-1$ succeeds} \\
a(t-1) + 1 & \text{otherwise.}
\end{cases}
\end{align*}

Let $a^{\mathcal{P}}_{m}(t)$ be the AoI of source $m$ in time-slot $t$ under a given {\color{black}policy} $\mathcal{P}$, and let $a^*_m(t)$ be the corresponding AoI under the oracle policy, discussed in detail section~\ref{ssec:Oracle}. We define the AoI regret at time $T$ as the cumulative difference in expected AoI for the two policies in time-slots 1 to $T$ summed over all the sources. Hence, the lesser the AoI a policy accumulates, the better it is. All policies operate under Assumption~\ref{assumption:initialConditions}, similar to the initial conditions in \cite{bhandari2020age}.

\begin{definition}[Age-of-Information Regret (AoI Regret)]
	\label{defn:regret}
	AoI regret under {\color{black}policy} $\mathcal{P}$ is denoted by $\mathcal{R}_{\mathcal{P}}(T)$ and
\begin{align}
\label{eq:regret}
\mathcal{R}_{\mathcal{P}}(T) = \sum_{m=1}^{M}\sum_{t = 1}^{T} \E[a^{\mathcal{P}}_m(t) - a^*_m(t)].
\end{align}
\end{definition}
\begin{assumption}[Initial Conditions]
	\label{assumption:initialConditions}
	The system starts operating at time-slot $t = -\infty$, but for any policy, decision making begins at $t=1$. It does not use information from observations in time-slots $t \leq 0$ to make decisions in time-slots $t \geq 1$. 
\end{assumption}





\section{Analytical Results}
\label{sec:analytic}
In this section, we state and discuss our main theoretical results. We characterize our Oracle/Genie policy and prove its optimality. We also provide an upper-bound for the regret of a Upper Confidence Bound (UCB)~\cite{auer2002finite} based policy for multi-user and multi-armed setting, called Distributed\footnote{here distributed learning refers to decentralized learning via local estimates} Learning with Fairness (DLF)~\cite{OpportunisticSpectrum}.

\subsection{Oracle Policy}
We will consider two candidates for the Oracle policy, namely I.I.D and Symmetric $M$-Periodic policy described in Theorem~\ref{thm:iid} and Definition~\ref{defn:pol-type} respectively.
\begin{theorem}[I.I.D Policy]
Let $\mathcal{S}(k) = \{1,\cdots, k\}$ be the set of k channels with the highest success probabilities. Consider the set of all I.I.D policies, the policy that schedules the $M$ sources on a random permutation of the set of arms $\mathcal{S}(M)$ at every instant (uniformly at random) is optimal.
\label{thm:iid}
\end{theorem}

\begin{definition}[Symmetric $M$-Periodic Policy]
\label{defn:pol-type}  
Let $\mathcal{D}_{m}$(t) be an N-element set containing the number of times each channel is scheduled by the source $m$ up till time-period $t$, arranged in decreasing order. Any periodic policy $\mathcal{P}$, with period M, such that $$\forall m \in [M], \mathcal{D}_m(\tau) =  \mathcal{D}(\tau) \text{, for } \tau \in (t-M,t], \forall t >M $$ 

is termed as a \emph{Symmetric $M$-Periodic Policy}.
\end{definition}
All the symmetric $M$-periodic policies under our consideration are void of any collision between any two sources and uses only the best-$M$ arms. Thus, for all the symmetric $M$-periodic policies under our consideration, $\mathcal{D}(\tau)$ is an $M$-length vector. These two conditions are necessary for any optimal policy. We refer interested readers to Lemmas~\ref{thm:no-col} and~\ref{thm:best-M} in Section~\ref{sec:proof}. Next, in Definition~\ref{defn:oracle}, we provide a special case of the Symmetric $M$-Periodic policy called the Round-Robin Policy.  
\begin{definition}[Round-Robin Policy]
    \label{defn:oracle}
   For a problem instance $\instance$, consider the index set $\mathcal{I}$ which is a random permutation of the arms in the set $\mathcal{S}(M)$. Then, in time-slot $t$, the round-robin policy schedules a source $m$ on the channel $\mathcal{I}_{((m + t)\bmod{M}) + 1}$.
\end{definition}
A symmetric $M$-periodic policy can be uniquely characterized by the sequence ${ \mathcal{D}(1),\cdots,\mathcal{D}(M)}$. By definition, the round-robin policy  is a symmetric $M$-periodic policy, and satisfies the property that each $\mathcal{D}(i)$, for $i \in [M]$, contains only 1s and 0s, and $\mathcal{D}(M) = (1,1,\cdots,1)$. Finally, Theorems~\ref{thm:oracele-m2}-\ref{thm:rr-vs-iid} characterize the optimality of the Round-Robin policy under certain conditions. This leads us to Conjecture~\ref{conj:opt}, where we generalize the Oracle to be optimal when there are more than two simultaneous sources.

\begin{algorithm}[!b]
	\DontPrintSemicolon 
    \textbf{Initialize:} Set $\hat{\mu}_n^m=0$ to be the estimated success probability of Channel $n$, $T_n^m(1)=0\;n\in[N]$.\;
	\While{$1\leq t\leq N$}{
	Schedule an update to a link $n(t)$ such that\;
	$n(t) = ((m + t) \bmod M) + 1$,\;
	Receive reward $X_{n(t)}^m(t)\sim \text{Ber}(\mu_{n(t)}^m)$\;
	$\hat{\mu}_{n(t)}^m = X_{n(t)}^m(t)$,\;
	$T_{n(t)}^m(t) = 1$\;
	$t = t + 1$\;}

	\While{$t\geq N+1$}{
	Set index $k = ((m + t) \bmod M) + 1$\;
	Let the set $\mathcal{O}_k$ contain the $k$ arms with the $k$ largest values in \eqref{eq:slk}
	\begin{equation}
	    \hat{\mu}_n^m + \sqrt{\frac{2\log t}{T_n^m(t-1)}}
	    \label{eq:slk}
	\end{equation}
	Schedule an update to a link $n(t)$ such that $$ n(t)=\arg \min_{n\in \mathcal{O}_k}\hat{\mu}_n^m - \sqrt{\frac{2\log t}{T_n^m(t-1)}}$$
	
	\uIf{Channel Acquired}{
	Receive reward $X_{n(t)}^m(t)\sim \text{Ber}(\mu_{n(t)}^m)$\;
    
	$\hat{\mu}_{n(t)}^m=\frac{\hat{\mu}^m_{n(t)}\cdot T_{n(t)}^m(t-1)+X_{n(t)}^m(t)}{T_{n(t)}(t-1)+1}$\;
	$T_{n(t)}^m(t)=T_{n(t)}^m(t-1)+1$\;}
	\Else{
	$\hat{\mu}_{n(t)}^m= \hat{\mu}_{n(t)}^m$ \text{(no update)}\;
	$T_{n(t)}^m(t)=T_{n(t)}^m(t-1)$\;
	}
	$t=t+1$}
	    
	\caption{{ Distributed Learning Algorithm with Fairness for $M$ Sources and $N$ Channels at Source $m$ }(DLF)}
	\label{algo:DLF}
\end{algorithm}

\begin{theorem}[Optimality of Round-Robin Policy for $M=2$]
For any problem instance $(M,N,\boldsymbol \mu)$ such that $M = 2, N \geq M$, a policy that schedules source $m \in [M]$ on the channel $((m + t)\bmod{M}) + 1$ is optimal. 
\label{thm:oracele-m2}
\end{theorem}

\begin{theorem}[Round-Robin is the best Symmetric $M$-Periodic Policy]
\label{thm:rr-best-smp}
For any problem instance $\instance$, there does not exist a symmetric $M$-periodic policy with a larger expected total age-of-information value than the round-robin policy.
\end{theorem}

\begin{theorem}[I.I.D. vs. Round-Robin]
\label{thm:rr-vs-iid}
The round-robin policy has a smaller expected total age-of-information value than the I.I.D. policy.
\end{theorem}

\begin{conjecture}
Generalizing Theorem~\ref{thm:oracele-m2} (for $M > 2$), over all possible permutation of arms in the set $\mathcal{S}(M)$, the round-robin policy is optimal for any problem instance $\instance$.
\label{conj:opt}
\end{conjecture}
Without loss of generality, hereafter we will use set $\mathcal{S}(M)$ as our pre-decided index set $\mathcal{I}$, as described in Definition~\ref{defn:oracle}, in all the algorithms discussed in Section~\ref{sec:mainResults}. Further, the regret of all the policies is calculated with respect to the round-robin policy, that is, the round-robin policy is our Oracle. Standard single-source bandit policies try to pick the best-estimated arm. Since we have multiple sources with equal priority, we try to pick the ``$k^{th}$ best" arm, where $k$ changes in a round-robin fashion for each source, thus ensuring fairness. However, this can still lead to collisions, as the estimates of the channel means are independently maintained by each source, and no information is shared. Thus, we check at every instant if the channel is acquired by the source (based on our collision model), and update the mean estimates only when the source has access to the channel. The regret for any policy $\mathcal{P}$ arises if a source $m$, for an index $k$, chooses an arm other than the desired one (which is used by the oracle) or due to a collision, the source is unable to acquire the channel. There is a trade-off between exploration and the number of collisions, since exploiting a smaller pool of channels increases the likelihood of collisions and vice-versa.

\subsection{Distributed Learning with Fairness (DLF)}
DLF is an equal-priority multi-source and multi-channel policy based on UCB-1. This policy tries to mimic our Oracle policy while avoiding collisions, by estimating the channel $\mu$s through a two-step process---estimating the set of best-$M$ channels and then selecting the channel with $k^{th}$ highest value of $\mu$ from the set, where the index $k$ is determined by the user $m$ and current slot $t$. This is formally described in Algorithm~\ref{algo:DLF}.  

\begin{algorithm}[t]
	\DontPrintSemicolon
    \textbf{Initialize:} Set $\hat{\mu}_n^m=0$ to be the estimated success probability of Channel $n$, $T_n^m(1)=0\;n\in[N]$.\;
	\While{$t\geq 1$}{
	$\alpha^m_{n}(t)=\hat{\mu}_n^m T_n^m(t-1)+1$,\;
	$\beta^m_{n}(t)=(1-\hat{\mu}_n^m)T_n^m(t-1)+1$,\;
	For each $n\in[N]$, pick a sample $\hat{\theta}^m_n(t)$ of distribution,
	\begin{equation}
	    \hat{\theta}^m_n(t)\sim \text{Beta}(\alpha^m_{n}(t),\beta^m_{n}(t))
	    \label{eq:TS1}
	\end{equation}

	Set index $k = ((m + t) \bmod M) + 1$\;
	Schedule an update to a link $n(t)$ such that it is the arm with the $k^{th}$ largest value in \eqref{eq:TS1}.\;
	\uIf{Channel Acquired}{
	Receive reward $X_{n(t)}^m(t)\sim \text{Ber}(\mu_{n(t)}^m)$\;
    
	$\hat{\mu}_{n(t)}^m=\frac{\hat{\mu}^m_{n(t)}\cdot T_{n(t)}^m(t-1)+X_{n(t)}^m(t)}{T_{n(t)}(t-1)+1}$\;
	$T_{n(t)}^m(t)=T_{n(t)}^m(t-1)+1$\;}
	\Else{
	$\hat{\mu}_{n(t)}^m= \hat{\mu}_{n(t)}^m$ \text{(no update)}\;
	$T_{n(t)}^m(t)=T_{n(t)}^m(t-1)$\;
	}
	$t=t+1$}
	    
	\caption{{ Distributed Learning-based Thompson Sampling for $M$ Sources and $N$ Channels at Source $m$ }(DL-TS)}
	\label{algo:TS}
\end{algorithm}
\begin{algorithm}[t]
	\DontPrintSemicolon 
    \textbf{Initialize:} Set $\hat{\mu}_n^m=0$ to be the estimated success probability of Channel $n$, $T_n^m(1)=0\;n\in[N]$.\;
	\While{$t\geq 1$}{
	Let $E(t) \sim \text{Ber}\left(\min\left\{1,mn\frac{\log t}{t}\right\}\right)$.\;
	\uIf{$E(t) = 1$}{\textit{DLF:} Schedule an update on a channel chosen by the DLF policy given in Algorithm~\ref{algo:DLF}.}
	\Else{\textit{Thompson Sampling:} Schedule an update on a channel chosen by the DL-TS policy given in Algorithm~\ref{algo:TS}.\; 
	}
	\uIf{Channel Acquired}{
	Receive reward $X_{n(t)}^m(t)\sim \text{Ber}(\mu_{n(t)}^m)$\;
    
	$\hat{\mu}_{n(t)}^m=\frac{\hat{\mu}^m_{n(t)}\cdot T_{n(t)}^m(t-1)+X_{n(t)}^m(t)}{T_{n(t)}(t-1)+1}$\;
	$T_{n(t)}^m(t)=T_{n(t)}^m(t-1)+1$\;}
	\Else{
	$\hat{\mu}_{n(t)}^m= \hat{\mu}_{n(t)}^m$ \text{(no update)}\;
	$T_{n(t)}^m(t)=T_{n(t)}^m(t-1)$\;
	}
	$t=t+1$
	
	}
	    
	\caption{{ Distributed Learning-based Hybrid policy for $M$ Sources and $N$ Channels at Source $m$ }(DLH)}
	\label{algo:hyb}
\end{algorithm}

\begin{theorem}[Performance of DLF]
Consider any problem instance $\instance$, such that $\mu_{\text{min}} = \min\limits_{i=1:N}\mu_i > 0$, $\Delta = \min\limits_{ i,j \in [M]; i > j} \mu_i - \mu_{j}$ , and a constant $c$.
Then, for any sufficiently large T $>$ N, under Assumption~\ref{assumption:initialConditions},
\begin{align}
    \mathcal{R}_{\text{DLF}}(T){\leq}
	\frac{M^2}{\mu_{\text{min}}}{+}\frac{M^2 c \log T}{\mu_{\text{min}}} {\bigg[}1{+}(N -1)\hspace{-0.3em}{\left(\frac{8 \log T}{\Delta^{2}}{+}1{+}\frac{2\pi^{2}}{3}\right)}{\bigg]}. \nn
\end{align}
\label{thm:DLF}
\end{theorem}

From Theorem~\ref{thm:DLF}, we conclude that DLF scales as $\OO(M^2 N \log^2 T)$. The proof first characterizes the source wise regret as a function of the expected number of times a non-desired channel is scheduled or another source acquires the desired channel under the policy. The result then follows using known upper-bounds for this quantity, from \cite{OpportunisticSpectrum} for DLF. We elaborate on this in Section~\ref{sec:proof} through Lemmas~\ref{lemma:expectedAgeBound} and \ref{lemma:bound-results}.

\section{Our Policies}
\label{sec:mainResults}
In this section, we extend Thompson Sampling~\cite{thompson1933likelihood} to our setting. We also propose a new hybrid policy, as well as AoI-Aware variants for each of these policies, and compare their performances via simulations.

\subsection{AoI-Agnostic Policies}
\begin{definition}[AoI-Agnostic Policies]
    A policy is AoI-agnostic if, given past scheduling decisions and the number of successful updates sent via each of the $N$ channels in the past, it does not explicitly use the age of information of any source in a time-slot to make scheduling decisions.
\end{definition}

While age-of-information is a new metric, to devise AoI-agnostic policies one can use the myriad of policies used for Bernoulli rewards, most commonly the Upper Confidence Bound (UCB)~\cite{auer2002finite} and Thompson Sampling~\cite{thompson1933likelihood}. These policies can be readily applied to a one-source setting, i.e. $M=1$~\cite{bhandari2020age}. DLF \cite{OpportunisticSpectrum}, as described in Algorithm~\ref{algo:DLF}, is based on UCB and tries to mimic our Oracle policy while avoiding collisions. Similarly, we extend Thompson Sampling for our setting, and we term this new policy as \emph{Distributed Learning-based Thompson Sampling (DL-TS)}, given in Algorithm~\ref{algo:TS}. Further, we combine these two policies to propose a hybrid policy - \emph{Distributed Learning-based Hybrid policy (DLH)}, detailed in Algorithm~\ref{algo:hyb}.

\begin{algorithm}[t]
	\DontPrintSemicolon 
    \textbf{Initialize:} Set $\hat{\mu}_n^m=0$ to be the estimated success probability of Channel $n$, $T_n^m(1)=0\;n\in[N]$.\;
	\While{$t\geq 1$}{
	$\alpha^m_{n}(t)=\hat{\mu}_n^m T_n^m(t-1)+1$,\;
	$\beta^m_{n}(t)=(1-\hat{\mu}_n^m)T_n^m(t-1)+1$,\;
	Set index $k = ((m + t) \bmod M) + 1$\;
	Let $\text{\textit{limit(t)}}= k^{th} \min\limits_{n\in [N]}\frac{\alpha_{n}^m(t)+\beta_{n}^m(t)}{\alpha_{n}^m(t)}$\;
	\uIf {$a(t-1)>\text{limit(t)}$}{
		\textit{Exploit:} Select channel with the $k^{th}$ highest estimated success probability
	}
	\Else{
		\textit{Explore:}\;
    	For each $n\in[N]$, pick a sample $\hat{\theta}^m_n(t)$ of distribution,
    	\begin{equation}
    	    \hat{\theta}^m_n(t)\sim \text{Beta}(\alpha^m_{n}(t),\beta^m_{n}(t))
    	    \label{eq:TSaa1}
    	\end{equation}
    	
    	Schedule an update to a link $n(t)$ such that it is the arm with the $k^{th}$ largest value in \eqref{eq:TSaa1}.\;
    	}
	\uIf{Channel Acquired}{
	Receive reward $X_{n(t)}^m(t)\sim \text{Ber}(\mu_{n(t)}^m)$\;
    
	$\hat{\mu}_{n(t)}^m=\frac{\hat{\mu}^m_{n(t)}\cdot T_{n(t)}^m(t-1)+X_{n(t)}^m(t)}{T_{n(t)}(t-1)+1}$\;
	$T_{n(t)}^m(t)=T_{n(t)}^m(t-1)+1$\;}
	\Else{
	$\hat{\mu}_{n(t)}^m= \hat{\mu}_{n(t)}^m$ \text{(no update)}\;
	$T_{n(t)}^m(t)=T_{n(t)}^m(t-1)$\;
	}
	$t=t+1$}
	    
	\caption{{ AoI-Aware Distributed Learning-based Thompson Sampling for $M$ Sources and $N$ Channels at Source $m$ }(DL-TS-AA)}
	\label{algo:TS-AA}
\end{algorithm}
\begin{algorithm}[t]
	\DontPrintSemicolon 
    \textbf{Initialize:} Set $\hat{\mu}_n^m=0$ to be the estimated success probability of Channel $n$, $T_n^m(1)=0\;n\in[N]$.\;
	\While{$1\leq t\leq N$}{
	Schedule an update to a link $n(t)$ such that\;
	$n(t) = ((m + t) \bmod M) + 1$,\;
	Receive reward $X_{n(t)}^m(t)\sim \text{Ber}(\mu_{n(t)}^m)$\;
	$\hat{\mu}_{n(t)}^m = X_{n(t)}^m(t)$,\;
	$T_{n(t)}^m(t) = 1$\;
	$t = t + 1$\;}

	\While{$t\geq N+1$}{
	$\alpha^m_{n}(t)=\hat{\mu}_n^m T_n^m(t-1)+1$,\;
	$\beta^m_{n}(t)=(1-\hat{\mu}_n^m)T_n^m(t-1)+1$,\;
	Set index $k = ((m + t) \bmod M) + 1$\;
	Let $\text{\textit{limit(t)}}= k^{th} \min\limits_{n\in [N]}\frac{\alpha_{n}^m(t)+\beta_{n}^m(t)}{\alpha_{n}^m(t)}$\;
	\uIf {$a(t-1)>\text{limit(t)}$}{
		\textit{Exploit:} Select channel with the $k^{th}$ highest estimated success probability
	}
	\Else{
	
	Let the set $\mathcal{O}_k$ contain the $k$ arms with the $k$ largest values in \eqref{eq:slk-2}
	\begin{equation}
	    \hat{\mu}_n^m + \sqrt{\frac{2\log t}{T_n^m(t-1)}}
	    \label{eq:slk-2}
	\end{equation}
	Schedule an update to a link $n(t)$ such that $$ n(t)=\arg \min_{n\in \mathcal{O}_k}\hat{\mu}_n^m - \sqrt{\frac{2\log t}{T_n^m(t-1)}}$$
	}
	\uIf{Channel Acquired}{
	Receive reward $X_{n(t)}^m(t)\sim \text{Ber}(\mu_{n(t)}^m)$\;
    
	$\hat{\mu}_{n(t)}^m=\frac{\hat{\mu}^m_{n(t)}\cdot T_{n(t)}^m(t-1)+X_{n(t)}^m(t)}{T_{n(t)}(t-1)+1}$\;
	$T_{n(t)}^m(t)=T_{n(t)}^m(t-1)+1$\;}
	\Else{
	$\hat{\mu}_{n(t)}^m= \hat{\mu}_{n(t)}^m$ \text{(no update)}\;
	$T_{n(t)}^m(t)=T_{n(t)}^m(t-1)$\;
	}
	$t=t+1$}
	    
	\caption{{ AoI-Aware Distributed Learning Algorithm with Fairness for $M$ Sources and $N$ Channels at Source $m$ }(DLF-AA)}
	\label{algo:DLF-AA}
\end{algorithm}

Empirically, we observe that DLF exhibits a lesser number of collisions while being more exploratory, and DL-TS is more exploitative but leads to a higher number of collisions. The DLH policy shown in Algorithm~\ref{algo:hyb}, switches between DLF and DL-TS using a Bernoulli random variable to balance this trade-off. As the chance of a collision is higher in the initial time-slots, here the algorithm uses DLF with a higher probability than DL-TS. However, the probability of using DLF decreases with an increase in the number of time-slots elapsed. 


\subsection{AoI-Aware Policies}

In this section, we propose AoI-aware variants of the three policies discussed in the previous section. In the classical MAB with Bernoulli rewards, the contribution of a sub-optimal pull to the cumulative regret is upper bounded by one, but in AoI bandits it can be greater than one. Also, unlike the MAB, for AoI bandits, the difference between AoIs under a candidate policy and the {\color{black}Oracle} policy in a time-slot can be unbounded. This motivates the need to take the current AoI value into account when making scheduling decisions.

Intuitively, it makes sense to explore when AoI is low and exploit when AoI is high since the cost of making a mistake is much higher when AoI is high. We use this intuition to design AoI-aware policies. The key idea behind these policies is that they mimic the original policies when AoI is below a threshold and exploit when AoI is equal to or above a threshold, for an appropriately chosen instance-dependent dynamic threshold. This threshold must be chosen carefully, as if the threshold is too small, we will exploit prematurely and thereby, incurring high regret. However, if the threshold is too high, only the original policy will be invoked, resulting in the same regret.

In each policy, at a source $m$, we maintain an estimate of success probability of the arms, denoted by $\boldsymbol{\hat\mu}^m$. For an index $k$ at time $t$, we mimic the original policy if the AoI is not more than $\frac{1}{\hat\mu^m_k}$ (which is the AoI value if the $k^{th}$ arm was used throughout). Otherwise, we exploit the ``$k^{th}$ best" arm (based on past observations) as per the source. Emperically, we observed that a dynamic threshold dependent on the estimate $\boldsymbol{\hat\mu}^m$ outperformed a constant (hard) threshold The AoI-Aware modifications to DLF, DL-TS and DLH are given in Algorithms~\ref{algo:DLF-AA}, \ref{algo:TS-AA} and \ref{algo:hyb-AA} respectively.

\begin{algorithm}[t]
	\DontPrintSemicolon 
    \textbf{Initialize:} Set $\hat{\mu}_n^m=0$ to be the estimated success probability of Channel $n$, $T_n^m(1)=0\;n\in[N]$.\;
	\While{$t\geq 1$}{
	Let $E(t) \sim \text{Ber}\left(\min\left\{1,mn\frac{\log t}{t}\right\}\right)$.\;
	\uIf{$E(t) = 1$}{\textit{DLF:} Schedule an update on a channel chosen by the DLF-AA policy given in Algorithm~\ref{algo:DLF-AA}.}
	\Else{\textit{Thompson Sampling:} Schedule an update on a channel chosen by the DL-TS-AA policy given in Algorithm~\ref{algo:TS-AA}.\; 
	}
	\uIf{Channel Acquired}{
	Receive reward $X_{n(t)}^m(t)\sim \text{Ber}(\mu_{n(t)}^m)$\;
    
	$\hat{\mu}_{n(t)}^m=\frac{\hat{\mu}^m_{n(t)}\cdot T_{n(t)}^m(t-1)+X_{n(t)}^m(t)}{T_{n(t)}(t-1)+1}$\;
	$T_{n(t)}^m(t)=T_{n(t)}^m(t-1)+1$\;}
	\Else{
	$\hat{\mu}_{n(t)}^m= \hat{\mu}_{n(t)}^m$ \text{(no update)}\;
	$T_{n(t)}^m(t)=T_{n(t)}^m(t-1)$\;
	}
	$t=t+1$
	
	}
	    
	\caption{{ AoI-Aware Distributed Learning-based Hybrid Policy for $M$ Sources and $N$ Channels at Source $m$ }(DLH-AA)}
	\label{algo:hyb-AA}
\end{algorithm}

\section{Simulations}
In this section, we will present the simulation results for all the six policies discussed in section~\ref{sec:mainResults}. All the simulations hereafter are conducted for $T=2\times10^4$ time-slots with each data point averaged across 200 iterations. For the purpose of simulating the policies, we choose $\boldsymbol{\mu}$ such that consecutive elements are equidistant, and the difference is denoted by $\Delta$. 

\begin{figure}[t!]
\centering
\begin{subfigure}[t]{\columnwidth}
\includegraphics[trim=3 5 5 3,clip,scale=0.4]{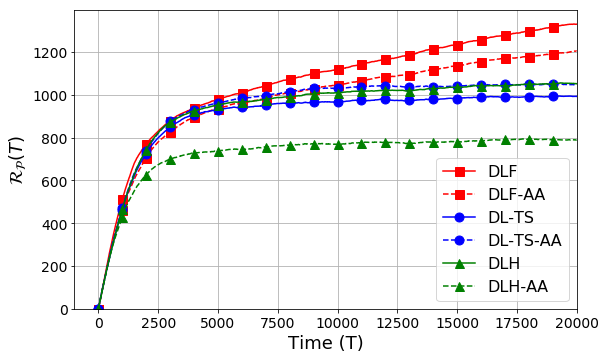}
\caption{$M=2,N=4,\boldsymbol{\mu}=[0.8,0.75,0.7,0.65]$}
\label{fig:regret-m-2}
\end{subfigure}
\\
\begin{subfigure}[t]{\columnwidth}
\includegraphics[trim=3 5 5 3,clip,scale=0.4]{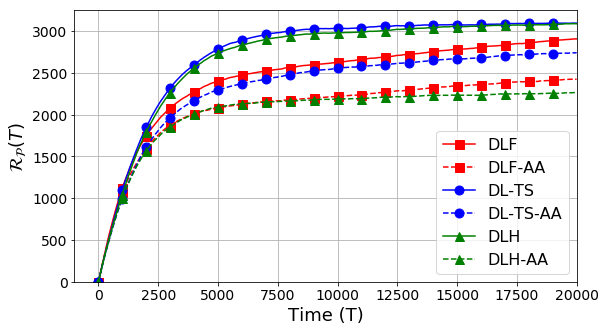}
\caption{$M=3,N=5,\boldsymbol{\mu}=[0.8,0.75,0.7,0.65, 0.6]$}
\label{fig:regret-m-3}
\end{subfigure}
\\
\begin{subfigure}[t]{\columnwidth}
\includegraphics[trim=3 5 5 3,clip,scale=0.4]{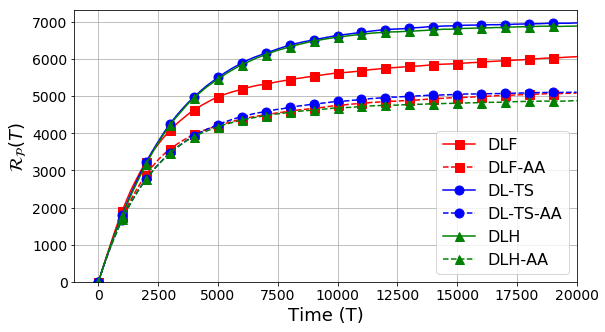}
\caption{$M=4, N=7,\boldsymbol{\mu}=[0.8,0.75,0.7,0.65,0.6,0.55,0.5]$}
\label{fig:regret-m-4}
\end{subfigure}
\caption{Regret plots for all six policies for $T=2\times10^4$ averaged over 200 iterations}
\label{fig:regrets}
\end{figure}

Figure~\ref{fig:regrets} shows the regret, $\mathcal{R}_{\calP}(T)$, for all the six policies mentioned previously for different problem instances $\instance$. We notice \emph{AoI-Aware} policies in all cases (except for DL-TS and DL-TS-AA in Figure~\ref{fig:regret-m-2}) have a smaller regret (and fewer collisions) than their \emph{AoI-Agnostic} counterparts. 


In Table~\ref{table:arms}, we compare the distribution of source-wise channels scheduled under two policies: DLF-AA and DLTS-AA for problem instance $(M=2,N=4,\boldsymbol{\mu}=[0.8,0.75,0.7,0.65])$. We observe that the number of pulls of sub-optimal arms, that is, arms $\notin \mathcal{S}(2)$ is higher in DLF-AA as compared to DLTS-AA. At the same time, the former policy averages about 414 collisions overall, whereas the latter averages about 556 collisions overall. Similarly in Table~\ref{table:arms-2}, we show the channel distribution for the problem instance $(M=2,N=4,\boldsymbol{\mu}=[0.8,0.75,0.7,0.65, 0.6])$. The average number of collisions for DLF-AA is 1071 and the same for DL-TS-AA is 1478.

\begin{table}[t]
	\begin{center}
	\caption{Distribution of Channels scheduled across sources for various policies when $M=2, N=4 \text{ and } \boldsymbol{\mu} = [0.8;0.65]$}
		\addtolength{\tabcolsep}{-1.5pt}
		\begin{tabular}{c|ccccc} 
			\hline
			\multicolumn{2}{c}{}&$\boldsymbol{n=1}$ & $\boldsymbol{n=2}$  & $\boldsymbol{n=3}$ & $\boldsymbol{n=4}$  \\
			\hline
			\hline
			\multirow{2}{*}{\textbf{DLF-AA}}& $\boldsymbol{m=1}$  & 9825 & 7429 & 1914 &  832 \\  
			& $\boldsymbol{m=2}$  & 9823 & 7421 & 1917 & 839 \\
            \hline
            \hline
            \multirow{2}{*}{\textbf{DLTS-AA}} &$\boldsymbol{m = 1}$ & 9871  & 9308 & 672 & 149 \\
			&$\boldsymbol{m=2}$ & 9879 & 9411 &  554 &  156 \\
			\hline
		\end{tabular}
		
		\label{table:arms}
	\end{center}
\end{table}
\begin{table}[t!]
	\begin{center}
	\caption{Distribution of Channels scheduled across sources for various policies when $M=3, N=5 \text{ and } \boldsymbol{\mu} = [0.8;0.6]$}
		\addtolength{\tabcolsep}{-1.5pt}
		\begin{tabular}{c|ccccccc} 
			\hline
			\multicolumn{2}{c}{}&$\boldsymbol{n=1}$ & $\boldsymbol{n=2}$  & $\boldsymbol{n=3}$ & $\boldsymbol{n=4}$ & $\boldsymbol{n=5}$ \\
			\hline
			\hline
			\multirow{2}{*}{\textbf{DLF-AA}}& $\boldsymbol{m=1}$ & 6621 & 6543 & 4598 & 1511 & 727 \\  
			& $\boldsymbol{m=2}$ & 6627 & 6524 & 4631 & 1487 &   731 \\
			& $\boldsymbol{m=3}$ & 6624 & 6535 & 4634 & 1479 & 728 \\
            \hline
            \hline
            \multirow{2}{*}{\textbf{DLTS-AA}} &$\boldsymbol{m=1}$ & 6640 &  6524 &  6023 & 655  & 158 \\
			&$\boldsymbol{m=2}$ & 6585 & 6557 & 6039 &  644 & 175 \\
			&$\boldsymbol{m=3}$& 6581 & 6573 & 6096 &  580  & 170 \\
			\hline
		\end{tabular}
		
		\label{table:arms-2}
	\end{center}
\end{table}

The same trend occurs when comparing DLF with DL-TS. This affirms that a trade-off exists between the two causes of regret: sub-optimal arm pulls and the number of collisions. This served as the primary motivation behind the hybrid policies DLH and DLH-AA, which lead to an intermediate number of both sub-optimal arm pulls and number of collisions. As one can see from Figure~\ref{fig:regrets}, DLH-AA performs the best in terms of cumulative regret for different values of M and N, especially when $\Delta$ is small. For higher $\Delta$, it was observed that the DL-TS policies performed better, since $\boldsymbol{\mu}$s being farther apart will lead to lesser collisions. DLH always resulted in intermediate values of regret, in between regret values for DLF and DL-TS. However, DLH is `closer' to DL-TS as compared to DLF, which is justified by the construction of the policy.

In both Tables~\ref{table:arms} and \ref{table:arms-2}, there is a sizable difference in the pulls of the $M^{th}$ arm as compared to the others in the set $\mathcal{S}(M)$. We believe that this is because across polices a majority of the sub-optimal pulls arise from scheduling the $M^{th}$ channel. This trend can be seen in Figure~3(b) of \cite{OpportunisticSpectrum}. Note that although the order of magnitude of the difference in the arm pulls is similar to ours, the absolute number is higher as their simulations are conducted for $T=10^6$ which is much larger than what we used. On increasing $T$, we notice that the number of arm pulls is close to the numbers presented in \cite{OpportunisticSpectrum}. 

\begin{figure}[t!]
\centering
\begin{subfigure}[t]{\columnwidth}
\includegraphics[trim=3 5 5 3,clip,scale=0.4]{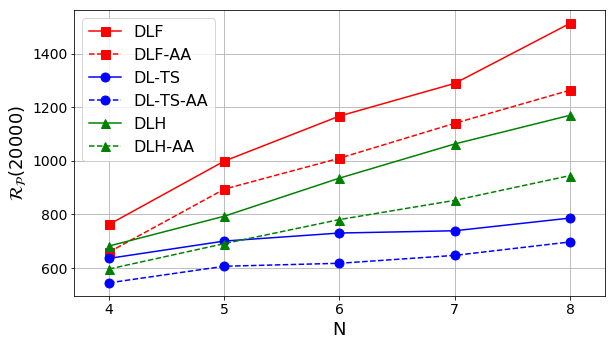}
\caption{$\mathcal{R}_{\mathcal{P}}(2\times10^4)$ vs. $N$; $M=3$, $\mu_1 = 0.9 \text{ and } \Delta =0.1$.}
\label{fig:var-n}
\end{subfigure}
\\
\begin{subfigure}[t]{\columnwidth}
\includegraphics[trim=3 5 5 3,clip,scale=0.4]{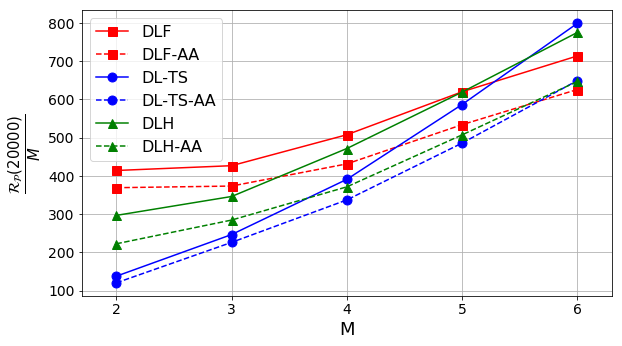}
\caption{ $\mathcal{R}_{\mathcal{P}}(2\times10^4)/M$ vs.  $M$; $N=7$ and $\boldsymbol{\mu}=[0.9;0.3]$}
\label{fig:var-m}
\end{subfigure}
\\
\begin{subfigure}[t]{\columnwidth}
\includegraphics[trim=3 5 5 3,clip,scale=0.4]{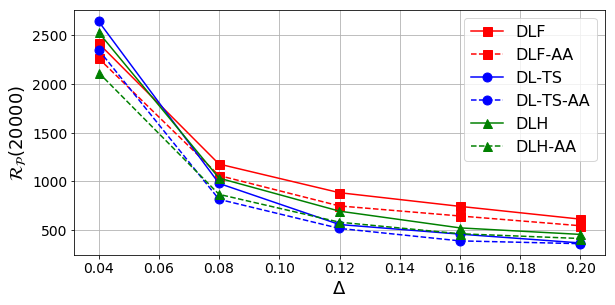}
\caption{$\mathcal{R}_{\mathcal{P}}(2\times10^4)$ vs. $\Delta$; $M=3, N=5 \text{ and } \mu_1=0.9$}
\label{fig:var-delta}
\end{subfigure}
\caption{Variation in regret, $\mathcal{R}_{\calP}(T)$, at $T=2\times10^4$, with different values of $N, M \text{ and } \Delta \text{ (or } \boldsymbol{\mu}$), in that order.}
\label{fig:variation}
\end{figure}

Figure~\ref{fig:variation} shows the variation of regret values of policies at $T=2\times10^4$ with different parameters. In Figure~\ref{fig:var-n}, we observe that the regret scales approximately linearly with $N$, with the DLF policies having the highest values of both slope and regret, and DL-TS the lowest. This is in accordance with the degree of exploration in the policies. For variation with $M$, in Figure~\ref{fig:var-m}, we plot $\mathcal{R}_{\calP}(T)/M$ instead of regret values, and since these plots are almost linear, suggesting that regret scales with $M^2$. Notice that the order of the magnitude of the slopes is reversed as compared to the previous plot---this is because as $M$ increases, DL-TS becomes more susceptible to collisions, thus accumulating higher regret. Figure~\ref{fig:var-delta} indicates an inverse relationship of regret $\mathcal{R}_\calP(T)$ with $\Delta$ in all the six policies. We note that, when $\Delta=0.04$, DLH and DLF based policies do better than the DL-TS based ones, DLH-AA being the best policy. This trend changes to DL-TS based policies having the lowest regret and the DLF based policies have the highest, with increasing $\Delta$.

\section{Proofs}
\label{sec:proof}
\subsection{Proof of Oracle}
\label{ssec:Oracle}
An optimal policy minimizes the expected value of the total age-of-information across all sources at any large and finite time-slot $T$. We first establish that any candidate for an optimal policy must satisfy two properties given in Lemmas~\ref{thm:no-col} and \ref{thm:best-M}. Then, in Lemma~\ref{lemma:effective-aoi}, we characterize the AoI for any I.I.D policy. Finally, using these lemmas and the KKT conditions, we can prove Theorem~\ref{thm:iid}.
\begin{lemma}[Collision Sub-Optimality]
\label{thm:no-col}
For a problem instance $\instance$, where $\mu_{\text{min}} = \min\limits_{i=1:N}\mu_i > 0$, any policy that leads to a collision in any time slot is sub-optimal.
\end{lemma}

\begin{proof}[Proof of Lemma~\ref{thm:no-col}]
This can be proved by contradiction. \\
Assume, a policy $\mathcal{P}$, which is optimal and there is a collision at channel $j \in [N]$ in time-slot $t$. We know that in the problem instance $\instance$, $N \geq M$. For a collision to happen, at least two sources have been scheduled on the channel $j$. Let the number of sources involved in the collision be $l\text{, and } 2 \leq l \leq M$ \\
$\therefore \exists i_{1}, \cdots ,i_{l-1} \in [N] \text{ and } i_{1} \neq i_{2} \neq \dots \neq i_{l-1} \neq j$, such that channels $i_{1},\cdots,i_{l-1}$ are idle. \\
$\mu_{\text{min}} = \min\limits_{k=1:N}\mu_k > 0 \implies \mu_{i_1}, \cdots, \mu_{i_{l-1}} > 0$ \\
Construct an alternate policy $\mathcal{P}^{'}$, which in time slot $t$, is identical to policy $\mathcal{P}$ for non-colliding sources, and for all but one sources colliding at channel $j$, schedules them on channels $i_1,\cdots, i_{l-1}$ randomly without repetition. \\
For these sources, under policy $\mathcal{P}$, the AoI would always increase by 1 , but under policy $\mathcal{P^{'}}$, the AoI would increase by 1 $wp. \leq 1 - \mu_{\text{min}} < 1$. 
$$\therefore \E[a_m^{\mathcal{P}}] \geq \E[a_m^{\mathcal{P'}}] \implies \sum_{m=1}^{M}\E[a_m^{\mathcal{P}}] \geq \sum_{m=1}^{M}\E[a_m^{\mathcal{P'}}],$$ which is a contradiction. No such optimal policy $\mathcal{P}$ exists.
\end{proof}
\begin{lemma}[Best-$M$ Channels]
\label{thm:best-M}

Let $\mathcal{S}(k) = \{1,\cdots, k\}$ be the set of k channels with the highest success probabilities.Then, any policy, that is not sub-optimal as per Lemma~\ref{thm:no-col}, but selects an arm not in $\mathcal{S}(M)$, is sub-optimal.

\end{lemma}
\begin{proof}[Proof of Lemma~\ref{thm:best-M}]
This can be proved by contradiction. \\
Assume, a policy $\mathcal{P}$, which is optimal such that in time-slot $t$, a source $m$ is scheduled on channel $j \in [N] \text{ where } j > M$.\\
From Lemma~\ref{thm:no-col}, in any time-slot t, M channels must be in use. By construction, it follows that $ \exists i \in [M] \text{ and } i < j$, such that channel $i$ is idle. Further,\\
$$\mu_{\text{min}} = \min\limits_{k=1:N}\mu_k > 0 \implies \mu_i \geq  \mu_j.$$
Construct an alternate policy $\mathcal{P}^{'}$, which in time slot $t$, is identical to policy $\mathcal{P}$ for all sources but m, where it schedule source m on channel $i$. For this source(s), under policy $\mathcal{P}$, the AoI would increase by 1 $wp. \text{ } 1 - \mu_{j}$, but under policy $\mathcal{P^{'}}$, the AoI would increase by 1 $wp.\text{ } 1 - \mu_{i} \leq 1 - \mu_{j}$. \\
$$\therefore \E[a_m^{\mathcal{P}}] \geq \E[a_m^{\mathcal{P'}}] \implies \sum_{m=1}^{M}\E[a_m^{\mathcal{P}}] \geq \sum_{m=1}^{M}\E[a_m^{\mathcal{P'}}], $$ which is a contradiction. No such optimal policy $\mathcal{P}$ exists.
\end{proof}
\begin{lemma}[AoI for I.I.D Policy] For any iid. policy $\calP$, that schedules channels in an iid. manner across time, let the effective success probability for a source $m$ be $\tilde{\mu}_m$. $\tilde{\mu}_m = \sum_{j=1}^{N} \lambda_{m,j} \mu_j$, where, $\lambda_{m,j}$ is the probability that source $m$ is scheduled on the channel $j$. Then, the expected value of the AoI of a source $m$ in time slot $t$ under this policy  is
    $\E[a^{\calP}_m(t)] = \frac{1}{\tilde{\mu}_m}$.
\label{lemma:effective-aoi}
\end{lemma}
\begin{proof}[Proof of Lemma~\ref{lemma:effective-aoi}]
By definition,
	$$
	\PP(a_m^{\mathcal{P}}(t) > \tau) = \prod_{i=0}^{\tau} \left(1-\mu_{k_m(t-i)}\right).
	$$	
	Note that since $a_m^{\calP}(t) \geq 1$ for all $t$ and m,
	$$
	\E[a_m^{\mathcal{P}}(t)] = \sum_{\tau=0}^{\infty} \PP(a_m(t) > \tau).
	$$
It follows that,
	\begin{align}
	\label{eq:aoi-defn}
	\E[a_m^{\mathcal{P}}(t)]= \E[\E[a_m^{\mathcal{P}}(t)]] =& \E\left[\sum_{\tau=0}^{\infty} \PP(a_m^{\mathcal{P}}(t) > \tau)\right] \nn\\
	=&\E\left[ \sum_{\tau=0}^{\infty}\prod_{i=0}^{\tau} \left(1-\mu_{k_m(t-i)}\right)\right] \\
	=&\sum_{\tau=0}^{\infty}\prod_{i=0}^{\tau} \bigg(1-\E[\mu_{k_m(t-i)}]\bigg) \nn \\
	=&\sum_{\tau=0}^{\infty}\prod_{i=0}^{\tau} \bigg(1-\tilde{\mu}_{m}\bigg) \nn \\
	=&\frac{1}{\tilde{\mu}_{m}}\text{ (From Assumption~\ref{assumption:initialConditions})}. \nn
	\end{align}
\end{proof}
\begin{proof}[Proof of Theorem~\ref{thm:iid}]
From Lemma~\ref{lemma:effective-aoi}, the total AoI is 
\begin{align}
        \sum_{m=1}^{M} \E[a_m^{\mathcal{P}}(t)] &= \sum_{m=1}^{M}\frac{1}{\tilde{\mu}_{m}} \nn \\
        &= \sum_{m=1}^{M}\frac{1}{\sum_{j=1}^{N} \lambda_{m,j} \mu_j}.
\label{eq:minimize}
\end{align}
We need to minimize \eqref{eq:minimize}, under the constraints $\forall i\in[M]$, $\sum_{j=1}^N \lambda_{i,j}=1$ and $\lambda_{i,j}\geq0$. Further, from Lemmas~\ref{thm:no-col} and \ref{thm:best-M},
\begin{equation}\forall j \in [M], \sum_{i=1}^M \lambda_{i,j} = 1
\text{ and } \forall i \in [M], j > M, \lambda_{i,j} = 0. \label{eq:cond}\end{equation}
Let $\pmb{\lambda} = (\lambda_{1,1},\lambda_{1,2},\cdots, \lambda_{i,j},\cdots, \lambda_{M,M})$ denote the set of $M^2 \text{ non-zero } \lambda_{i,j}$, according to the constraints. Formalizing this constraint-minimization problem as per the KKT conditions we get,
$$ \text{Minimize } f(\pmb{\lambda}) = \sum_{m=1}^{M}\frac{1}{\sum_{j=1}^{N} \lambda_{m,j} \mu_j} $$
$$ g_{i,j}(\pmb{\lambda}) = - \lambda_{i,j} \leq 0 \text{ }\forall i,j \in [M]$$ 
$$ h^s_i(\pmb{\lambda}) = \sum_{j=1}^M \lambda_{i,j} - 1 = 0 \text{ }\forall i \in [M]$$
$$ h^c_j(\pmb{\lambda}) = \sum_{i=1}^M \lambda_{i,j} - 1 = 0 \text{ }\forall j \in [M] $$
Let $L(\pmb{\lambda}, \alpha, \beta) = f(\pmb{\lambda}) + \alpha^\top \boldsymbol{g}(\pmb{\lambda}) + \beta^\top \boldsymbol{h}(\pmb{\lambda}) $ be the Lagrangian function, where $\boldsymbol{g}(\pmb{\lambda}) = (g_{1,1}(\pmb{\lambda}),g_{1,2}(\pmb{\lambda}),\cdots,g_{M,M}(\pmb{\lambda}))$ and $\boldsymbol{h}(\pmb{\lambda}) = (h^s_1(\pmb{\lambda}),\cdots,h^s_M(\pmb{\lambda}), h^c_1(\pmb{\lambda}), \cdots, h^c_M(\pmb{\lambda})$. \\\\
Considering $\pmb{\lambda}^* = (\frac{1}{M}, \frac{1}{M}, \cdots,\frac{1}{M})_{M^2 \text{ times}}$ as a candidate optimal point, we show that the necessary conditions are verified.\\
\emph{Stationarity:} $\nabla f(\pmb{\lambda}^*) + \alpha^\top \nabla\boldsymbol{g}(\pmb{\lambda}^*) + \beta^\top \nabla\boldsymbol{h}(\pmb{\lambda}^*) = 0$. \\
Set $\alpha = \overline{0}$. 
$$\frac{\partial f(\pmb{\lambda})}{\partial \lambda_{i,j}} \bigg|_{\pmb{\lambda} = \pmb{\lambda}^*} = \frac{-\mu_{j} M^2}{\left( \sum_{j=1}^{M} \mu_j \right)^2}$$
$$\frac{\partial h^s_i(\pmb{\lambda})}{\partial \lambda_{i,j}} = 1, \frac{\partial h^c_j(\pmb{\lambda})}{\partial \lambda_{i,j}} = 1 $$
Let $\beta_i = \beta_{M + j} = \frac{\mu_{j} M^2}{\left( 2\sum_{j=1}^{M} \mu_j \right)^2}$, then the stationarity condition is satisfied.\\
\emph{Primal feasibility:} $$g_{i,j}(\pmb{\lambda}^*) = - \frac{1}{M} \leq 0 \because M>0 $$
$$h^s_i(\pmb{\lambda}^*) = \sum_{j=1}^M \frac{1}{M} - 1 = 0$$
$$h^c_j(\pmb{\lambda}^*) = \sum_{i=1}^M \frac{1}{M} - 1 = 0$$
\emph{Dual feasibility:} $$\alpha_{i,j}\geq 0\forall i,j \in [M] \because \alpha=\overline{0}$$
\emph{Complementary slackness:} $$\alpha^\top \boldsymbol{g}(\pmb{\lambda}^*) = 0 \because \alpha=\overline{0}$$

We know that $f(\pmb{\lambda})$ is a convex function, and $\boldsymbol{g}(\pmb{\lambda})$ comprise of linear and thus, convex functions, and $\boldsymbol{h}(\pmb{\lambda})$ comprises of affine functions. As a result, the necessary conditions are sufficient for optimality, giving us the desired result.
\end{proof}

We design two schedules: \emph{Policy $\calP_A$} and \emph{Policy $\calP_B$}. By enumerating all cases, we show that the AoI regret under \emph{Policy $\calP_A$} is less than that under \emph{Policy $\calP_B$}, and as a result prove Theorem~\ref{thm:oracele-m2}.

\begin{proof}[Proof of Theorem~\ref{thm:oracele-m2}]
Consider a problem instance $\instance$, where $M=2$ and a system of two coupled policies $\calP_A$ and $\calP_B$ at time $T = k_1 + 2 + k_2$, where $k_1 \geq 1\text{, }k_2 \geq 0$ and $T > 2$ \\\\
\textit{Policy $\calP_A$:} The round-robin policy that schedules source $m \in [M]$ on the channel $((m + t)\bmod{M}) + 1$ in time-slot $t$.\\\\
\textit{Policy $\calP_B$:} A policy that is not sub-optimal as per Lemmas~\ref{thm:no-col} and \ref{thm:best-M}, and makes the same scheduling decisions as Policy $\calP_A$ in time-slot $t \in [1,k_1] \cup [k_1 + 3,T]$. In $M$-length interval, $t \in [k_1 +1, k_1 + 2]$, the sources are scheduled on repeated arms from set $\mathcal{S}(2)$. 

Next, we need to show that: $$\sum_{m = 1}^{M} \E\big[a_m^{\calP_B}(T)\big] \geq \sum_{m = 1}^{M} \E\big[a_m^{\calP_A}(T)\big].$$\\
From \eqref{eq:aoi-defn}\footnote{The expressions underlined are same under policies $\calP_{\mathcal{A}}$ and $\calP_{\mathcal{B}}$},

\begin{align}
\label{eq:pol-a-b}
&\E\big[a_m^{\calP_B}(T)\big] \nn \\
=&\E\bigg[\sum_{\tau=0}^{\infty}\prod_{i=0}^{\tau} \left(1-\mu_{k_m^B(t-i)}\right)\bigg] \nn \\
=& 1 + \E\bigg[\sum_{\tau=0}^{k_2 - 1}\prod_{i=0}^{\tau}\left(1-\mu_{k_m^B(t-i)}\right) \nn \\
+&  \sum_{\tau=k_2}^{k_2 + M - 1}\prod_{i=0}^{\tau}  \left(1-\mu_{k_m^B(t-i)}\right) + \sum_{\tau=k_2 + M}^{T}\prod_{i=0}^{\tau}  \left(1-\mu_{k_m^B(t-i)}\right) \bigg] \nn \\
=& \underbracket[0.1pt]{1 + \E\bigg[\sum_{\tau=0}^{k_2 - 1}\prod_{i=0}^{\tau}\left(1-\mu_{k_m^A(t-i)}\right) \bigg] }\nn \\
+& \E \bigg[ \sum_{\tau=k_2}^{k_2 + M -1}\underbracket[0.1pt]{\prod_{i=0}^{k_2 - 1}  \left(1-\mu_{k_m^A(t-i)}\right)}\prod_{i=k_2}^{\tau}  \left(1-\mu_{k_m^B(t-i)}\right) \bigg] \nn \\
+&\E \bigg[\sum_{\tau=k_2 + M}^{T}\bigg(\underbracket[0.1pt]{\prod_{i=0}^{k_2 - 1}  \left(1-\mu_{k_m^A(t-i)}\right)}\prod_{i=k_2}^{k_2 + M-1}  \left(1-\mu_{k_m^B(t-i)}\right) \nn \\
\times& \underbracket[0.1pt]{\prod_{i=k_2+M}^{\tau}\left(1-\mu_{k_m^A(t-i)}\right)}\bigg)\bigg]. 
\end{align}
Using \eqref{eq:pol-a-b}, and $M=2$ we compute
\begin{align}
&\sum_{m = 1}^{M} \E\big[a_m^{\calP_B}(T)\big] - \sum_{m = 1}^{M} \E\big[a_m^{\calP_A}(T)\big] \nn \\
=&\sum_{m=1}^{2} \bigg\{ \sum_{\tau=k_2}^{k_2 + M -1}\prod_{i=0}^{k_2 - 1}  \left(1-\mu_{k_m^A(t-i)}\right)\bigg(\prod_{i=k_2}^{\tau}  \left(1-\mu_{k_m^B(t-i)}\right) \nn \\ 
-& \prod_{i=k_2}^{\tau}  \left(1-\mu_{k_m^A(t-i)}\right)\bigg) \nn +\sum_{\tau=k_2 + M}^{T}\bigg(\prod_{i=0}^{k_2 - 1}  \left(1-\mu_{k_m^A(t-i)}\right)\nn \\
\times&\bigg( \prod_{i=k_2}^{k_2 + M-1}  \left(1-\mu_{k_m^B(t-i)}\right) - \prod_{i=k_2}^{k_2 + M-1}  \left(1-\mu_{k_m^A(t-i)}\right) \bigg) \nn \\
\times& \prod_{i=k_2 + M}^{\tau}  \left(1-\mu_{k_m^A(t-i)}\right)\bigg) \bigg\}.
\label{eq:diff-a-b}
\end{align}
The first term in \eqref{eq:diff-a-b} can be simplified as,
\begin{align}
&\sum_{m=1}^{2} \bigg\{ \sum_{\tau=k_2}^{k_2 + M -1}\prod_{i=0}^{k_2 - 1}  \left(1-\mu_{k_m^A(t-i)}\right)\bigg(\prod_{i=k_2}^{\tau}  \left(1-\mu_{k_m^B(t-i)}\right) \nn \\
-& \prod_{i=k_2}^{\tau}  \left(1-\mu_{k_m^A(t-i)}\right)\bigg) \bigg\}.  \nn
\end{align}
\textit{Case-I: If $k_2 \bmod 2 = 0$}
\begin{align}
=&\prod_{i=0}^{k_2 - 1}  \bigg((1-\mu_{1})(1-\mu_{2})\bigg)^\frac{k_2}{2} \nn \\
\times& \bigg(\left(1 - \mu_1\right)^2 + \left(1 - \mu_2\right)^2 - 2\left(1 - \mu_1\right)\left(1 - \mu_2\right) \bigg) \geq& 0.
\label{eq:case-1}
\end{align}
\textit{Case-II: If $k_2 \bmod 2 = 1$}
\begin{align}
=& \prod_{i=0}^{k_2 - 1}  \bigg((1-\mu_{1})(1-\mu_{2})\bigg)^\frac{k_2 - 1}{2}\bigg(\left(1 - \mu_1\right)^2 + \left(1 - \mu_2\right)^2 \nn \\
-& 2\left(1 - \mu_1\right)\left(1 - \mu_2\right)\bigg) + \prod_{i=0}^{k_2 - 1}  \bigg((1-\mu_{1})(1-\mu_{2})\bigg)^\frac{k_2 - 1}{2} \nn \\
\times& \bigg(\left(1 - \mu_1\right)^3 + \left(1 - \mu_2\right)^3 - \left(1 - \mu_1\right)\left(1 - \mu_2\right)\left(\mu_1 - \mu_2\right) \bigg) \nn \\
\geq& \prod_{i=0}^{k_2 - 1}  \bigg((1-\mu_{1})(1-\mu_{2})\bigg)^\frac{k_2 - 1}{2}\bigg(\left(\mu_1 - \mu_2\right)^2\left(\mu_1 + \mu_2\right) \bigg) \nn \\
\geq& 0.
\label{eq:case-2}
\end{align}
\eqref{eq:case-1} and \eqref{eq:case-2} show that the first term in \eqref{eq:diff-a-b} is non-negative. A similar proof can be shown for the second term as well. This gives the desired result ($M=2$)-
$$ \sum_{m = 1}^{M} \E\big[a_m^{\calP_B}\big] - \sum_{m = 1}^{M} \E\big[a_m^{\calP_A}\big] \geq 0. $$

\end{proof}
\subsection{Proof of Theorems~\ref{thm:rr-best-smp} and~\ref{thm:rr-vs-iid}}
We then, generalize the proof of Theorem~\ref{thm:oracele-m2} and use the Muirhead's inequality~\cite{bullen2013handbook} to prove Theorems~\ref{thm:rr-best-smp} and \ref{thm:rr-vs-iid}. The structure of these proofs is similar to the above proof.
\begin{proof}[Proof of Theorem~\ref{thm:rr-best-smp}]
Consider a problem instance $\instance$, and a system of two coupled policies $\calP_A$ and $\calP_B$ at time $T = k_1 + M + k_2$, where $k_1 \geq 1\text{, }k_2 \geq 0$ and $T > M$ \\\\
\textit{Policy $\calP_A$:} Consider the index set $\mathcal{I}$ which is a random permutation of the arms in the set $\mathcal{S}(M)$. Policy $\calP_A$ is the  round-robin policy that schedules source $m \in [M]$ on the channel $\mathcal{I}_{((m + t)\bmod{M}) + 1}$ in time-slot $t$.\\\\
\textit{Policy $\calP_B$:} A policy, adhering to Lemmas~\ref{thm:no-col} and \ref{thm:best-M}, that makes the same scheduling decisions as Policy $\calP_A$ in time-slot $t \in [1,k_1] \cup [k_1 + M +1,T]$. In the $M$-length interval $t \in [k_1 +1, k_1 + M]$, the channels are scheduled as per a symmetric $M$-periodic policy, characterized by the sequence $\mathcal{D}^\mathcal{B}(1),\cdots,\mathcal{D}^\mathcal{B}(M)$. 

Both the policies are randomized, and constitute different realizations. Thus, the age-of-information accrued is captured by the expectation over all the random instances possible within each policy. Next, we need to show that: $$\sum_{m = 1}^{M} \E\big[a_m^{\calP_B}(T)\big] \geq \sum_{m = 1}^{M} \E\big[a_m^{\calP_A}(T)\big].$$\\
For a source $m$, using \eqref{eq:pol-a-b} we get that,
\begin{align}
\label{eq:general-diff}
   &\E\big[a_m^{\calP_B}(T)\big] -  \E\big[a_m^{\calP_A}(T)\big]   \nn \\
   =& \sum_{\tau=k_2}^{k_2 + M -1}\E \bigg[ \prod_{i=0}^{k_2 - 1}  \left(1-\mu_{k_m^A(t-i)}\right)\bigg\{ \prod_{i=k_2}^{\tau}  \left(1-\mu_{k_m^B(t-i)}\right) \nn \\
   -& \prod_{i=k_2}^{\tau}  \left(1-\mu_{k_m^A(t-i)}\right) \bigg \} \bigg] \nn \\
+&\sum_{\tau=k_2 + M}^{T}\E \bigg[\prod_{i=0}^{k_2 - 1}  \left(1-\mu_{k_m^A(t-i)}\right)\prod_{i=k_2 + M}^{\tau}\left(1-\mu_{k_m^A(t-i)}\right)  \nn \\
\times& \bigg \{ \prod_{i=k_2}^{k_2 + M-1}  \left(1-\mu_{k_m^B(t-i)}\right) - \prod_{i=k_2}^{k_2 + M-1}  \left(1-\mu_{k_m^A(t-i)}\right) \bigg \} \bigg]
\end{align}
Let us consider each of the terms inside the first summation in \eqref{eq:general-diff}. By construction, $\mathcal{D}^\mathcal{B}(\tau)$ majorizes over  $\mathcal{D}^\mathcal{A}(\tau)$. The expectation, for each policy, is equivalent to the Muirhead mean ([$\mathcal{D}(\tau)$]-mean) of the positive real numbers $1-\mu_1, 1-\mu_2, \cdots, 1-\mu_M$. Thus, by Muirhead's inequality, each of the expectation terms is positive. Similary, we can argue the same for expectation terms inside the second summation, resulting in 
\begin{align}
\label{eq:thm-smp}
\E[a_m^{\calP_B}(T)\big] -  \E\big[a_m^{\calP_A}(T)\big] \geq 0.
\end{align}
Summing for all sources we get the result, $$\sum_{m = 1}^{M} \E\big[a_m^{\calP_B}(T)\big] \geq \sum_{m = 1}^{M} \E\big[a_m^{\calP_A}(T)\big].$$
\end{proof}

\begin{proof}[Proof of Theorem~\ref{thm:rr-vs-iid}]
Consider a problem instance $\instance$, and a system of two coupled policies $\calP_A$ and $\calP_B$ at time $T = k_1 + M + k_2$, where $k_1 \geq 1\text{, }k_2 \geq 0$ and $T > M$ \vspace{0.5 em}\\
\textit{Policy $\calP_A$:} Consider the index set $\mathcal{I}$ which is a random permutation of the arms in the set $\mathcal{S}(M)$. Policy $\calP_A$ is the  round-robin policy that schedules source $m \in [M]$ on the channel $\mathcal{I}_{((m + t)\bmod{M}) + 1}$ in time-slot $t$.\vspace{0.5 em}\\
\textit{Policy $\calP_B$:} A policy, adhering to Lemmas~\ref{thm:no-col} and \ref{thm:best-M}, that makes the same scheduling decisions as Policy $\calP_A$ in time-slot $t \in [1,k_1] \cup [k_1 + M +1,T]$. In the $M$-length interval $t \in [k_1 +1, k_1 + M]$, the channels are scheduled uniformly and randomly in an i.i.d manner. 

Both the policies are randomized, and constitute different realizations. Thus, the age-of-information accrued is captured by the expectation over all the random instances possible within each policy. 

In case of policy $\calP_\mathcal{B}$, for an $l$-length interval, there are $M^l$ possible sequences of scheduled arms, and each sequence has the same probability. Suppose, $\mathcal{X}(l)$ represents the set of all possible sequences of length $l$ ( $\vert\mathcal{X}(l)\vert = M^l$ ). Let $\mathcal{X}^s(l)$ represent the set of sequences denoted by $\mathcal{D}^s(l)$, then,
\begin{align}
\label{eq:sets}
   \mathcal{X}(l) = \bigcup_{s} \mathcal{X}^s(l) \text{ and } \mathcal{X}^s_i(l) \cap \mathcal{X}^s_j(l) = \phi \text{ for } i \neq j. 
\end{align} Here the union extends over all the possible ways of constructing the vector $\mathcal{D}(l)$. Using \eqref{eq:thm-smp} and \eqref{eq:sets}, we get,
\begin{align}
\E\big[a_m^{\calP_B}(T)\big]
=& {\Big(\sum\limits_s \vert\mathcal{X}^s(T)\vert\cdot\E\big[a_m^{\calP_s}(T)\big]}\Big)\Big/{\vert\mathcal{X}(T)\vert} \nn \\
\geq& {\Big(\sum\limits_s \vert\mathcal{X}^s(T)\vert\cdot\E\big[a_m^{\calP_\mathcal{A}}(T)\big]}\Big)\Big/{\vert\mathcal{X}(T)\vert} \nn \\
=& \E\big[a_m^{\calP_\mathcal{A}}(T)\big]\cdot\sum\limits_s\vert\mathcal{X}^s(T)\vert\Big/{\vert\mathcal{X}(T)\vert} \nn \\
=& \E\big[a_m^{\calP_\mathcal{A}}(T)\big].
\end{align}
Summing for all sources we get the result, $$\sum_{m = 1}^{M} \E\big[a_m^{\calP_B}(T)\big] \geq \sum_{m = 1}^{M} \E\big[a_m^{\calP_A}(T)\big].$$
\end{proof}
\subsection{Proof of Theorem~\ref{thm:DLF}}
We first upper bound the expected cumulative source-wise AoI for any schedule by the expected cumulative source-wise AoI of an alternative schedule in which all uses of sub-optimal channels in the original schedule are replaced by all sources using (and thereby colliding on) the worst channel. Only one of these sources, at random, acquires the channel. We further upper bound the expected cumulative source-wise AoI with another schedule where all uses of the worst channel are clustered together starting from $T=1$, followed by the oracle allotments. We can express the resultant upper bound on the source-wise expected cumulative AoI in terms of the number of deviations from the oracle allotments and the number of times another source acquires the same channel as the source (resulting in a collision). This is effectively done in Lemma~\ref{lemma:expectedAgeBound}.

\begin{lemma}
	\label{lemma:expectedAgeBound}
	Let $k_m(t)$ denote the index of the communication channel used in time-slot $t$ by policy $\mathcal{P}$ and $k^*_m(t)$ be the index of the channel used by the Oracle, for source m. Let $\textbf{K}_m(T) = \{k_m(1), k_m(2), \cdots, k_m(T)\}$ be the sequence of channels used in time-slots $1$ to $T$ and
$$N_1(\textbf{K}_m(T)) = \sum_{t=1}^T  \mathbbm{1}_{k_m(t) \neq k^*_m(t)},$$
denote the number of time-slots in which a sub-optimal channel is used by source $m$ and 
$$N_2(\textbf{K}_m(T)) = \sum_{t=1}^T  \sum_{i\neq m} \mathbbm{1}_{k_i(t) = k^*_m(t)},$$ 
denote the number of time-slots the optimal arm (alloted by the oracle) was used by some other source $i \neq m$.
	Then, under Assumption \ref{assumption:initialConditions}, for a constant $c$, the regret for source m,
	\begin{align*}
	\mathcal{R}_{\mathcal{P}}^m(T)  \leq  \frac{M}{\mu_{\text{min}}} + \frac{M c \log T}{\mu_{\text{min}}} \left(1+ \E[N(\textbf{K}_m(T))]\right),
	\end{align*}
where, $N(\textbf{K}_m(T)) = N_1(\textbf{K}_m(T)) + N_2(\textbf{K}_m(T)).$
\end{lemma}
The next lemma summarizes the results from Theorem 1 in~\cite{OpportunisticSpectrum} to provide upper bounds on $ N_1(\textbf{K}_m(T))$ and $ N_2(\textbf{K}_m(T))$ as per the DLF policy.
\begin{lemma}
	\label{lemma:bound-results}
	Let $k_m(t)$ denote the index of the communication channel used in time-slot $t$ by policy $\mathcal{P}$ and $k^*_m(t)$ be the index of the channel used by the Oracle, for source m. Let $\displaystyle \E_\text{DLF} \left[N_1(\mathbf{K}(T))\right]$ denote the expected number of time-slots in which a sub-optimal channel is picked in time-slots 1 to $T$ by DLF. Let $\displaystyle \E_\text{DLF} \left[N_2(\mathbf{K}(T))\right]$ denote the expected number of time-slots in which the optimal arms was picked by some other source $i \neq m$ time-slots 1 to $T$ by DLF.
	Then, for $T>N$,
	\begin{align*}
	\E_\text{DLF} \left[N_1(\mathbf{K}(T))\right] &\leq (N-1)\left(\frac{8 \log T}{\Delta^{2}} + 1 + \frac{2\pi^{2}}{3} \right), \\
	\E_\text{DLF} \left[N_2(\mathbf{K}(T))\right] &\leq (M-1)\left(\frac{8 \log T}{\Delta^{2}} + 1 + \frac{2\pi^{2}}{3} \right),
	\end{align*}
	where $\Delta = \min\limits_{ i,j \in [M]; i > j} \mu_i - \mu_{j}$.
	
\end{lemma}
When calculating the overall regret from individual source-wise regrets, $\sum_{m=1}^{M} \mathcal{R}_{\mathcal{P}}^m$, we need to use lemma~\ref{lemma:expectedAgeBound}. However, on a closer look, one observes that for a source $m$ and an event under $N_2(\textbf{K}_m(T))$ are captured by $N_1(\textbf{K}_l(T))$ for some other source $l$. That is, $$ \sum_{m=1}^{M} N_2(\textbf{K}_m(T)) \subseteq \sum_{m=1}^{M} N_1(\textbf{K}_m(T)). $$ Thus, to avoid over-count, we capture these events only under $\sum_{m=1}^{M} N_1(\textbf{K}_m(T))$. This gives us 
\begin{align}
    \label{eq:n1n2}
    \mathcal{R}_{\mathcal{P}} & = \sum_{m=1}^{M} \mathcal{R}_{\mathcal{P}}^m \nn \\
    & \leq \sum_{m=1}^{M}  \frac{M}{\mu_{\text{min}}} + \frac{M c \log T}{\mu_{\text{min}}} \left(1+ \E[N_1(\textbf{K}_m(T))]\right).
\end{align}
\begin{remark}[Performance of DL-TS]
Lemma~\ref{lemma:expectedAgeBound} provides a generic result that shows the performance of any policy $\calP$. However, in the case of DL-TS, a generic result like Lemma~\ref{lemma:bound-results} is not available for Thompson Sampling. It is not straight-forward to adapt the proof given for Theorem 2 in \cite{kaufmann2012thompson} to our case, that is, periodic pulls of the $k^{th}$ best arm. Achieving the same remains to be an open problem.
\end{remark}

\begin{proof} [Proof of Lemma \ref{lemma:expectedAgeBound}]

$$\text{Let } \mu^* = \left(\sum_{i = 1}^{M} \mu_i \right) \big/ M.$$
By definition,
	$$
	\PP(a_m^{\mathcal{P}}(t) > \tau) = \prod_{i=0}^{\tau} \left(1-\mu_{k_m(t-i)}\right).
	$$	
	Note that since $a_m^{\calP}(t) \geq 1$ for all $t$ and m,
	$$
	\E[a_m^{\mathcal{P}}(t)] = \sum_{\tau=0}^{\infty} \PP(a_m(t) > \tau).
	$$
	It follows that,
	\begin{align}
	\E[a_m^{\mathcal{P}}(t)]= \E[\E[a_m^{\mathcal{P}}(t)]] =& \E\left[\sum_{\tau=0}^{\infty} \PP(a_m^{\mathcal{P}}(t) > \tau)\right] \nn\\
	=&\E\left[ \sum_{\tau=0}^{\infty}\prod_{i=0}^{\tau} \left(1-\mu_{k_m(t-i)}\right)\right].
	\label{eq:doubleExpectation}
	\end{align}

	For $t \geq c \log T$, we define $E_t$ as the event that $k_m(\tau) = k^*_m(\tau)$ for $t-c \log T + 1 \leq \tau \leq t$. Then,
	\begin{align*}
	\E &\left[ \sum_{\tau=0}^{\infty}\prod_{i=0}^{\tau} \left(1-\mu_{k_m(t-i)}\right)\bigg | E_t\right] \\
	&\leq \sum_{i=1}^{c \log T}\prod_{j=1}^{i}(1-\mu^{*}) \\
	&\hspace{.2in}+ \sum_{i=c \log T+1}^{\infty}
	 (1-\mu^{*})^{c \log T}\prod_{j=c \log T+1}^{i}(1-\frac{\mu_{\text{min}}}{M}).
	\end{align*}

For $t \geq M c \log T$, we define $E_t$ as the event that $k_m(\tau) = k^*_m(\tau)$ for $t-M c \log T + 1 \leq \tau \leq t$. Then,
\begin{align*}
	&\E\left[ \sum_{\tau=0}^{\infty}\prod_{i=0}^{\tau} \left(1-\mu_{k_m(t-i)}\right)\bigg | E_t\right] \\
	&\leq \sum_{i=1}^{M c \log T}\prod_{j=1}^{i}(1-\mu_{k_m^*}(t)) \\
	&\hspace{.05in}+ \sum_{i=M c \log T+1}^{\infty}
	 \left(\prod_{m=1}^{M}(1-\mu_m)\right)^{c \log T}\prod_{j= M c \log T+1}^{i}(1-\frac{\mu_{\text{min}}}{M}).
	\end{align*}
	Here, the worst case is that all sources collide at the arm with the lowest transmission probability. Based on our collision model, a source $m$ will acquire this channel with probability $1/M$.
	Further note that,
	\begin{align*}
	&\hspace{0.1in} \sum_{i=M c \log T+1}^{\infty}\left(\prod_{m=1}^{M}(1-\mu_m)\right)^{c \log T}\prod_{j=c \log T+1}^{i}(1-\frac{\mu_{\text{min}}}{M}) \\
	&\hspace{0.18in} \leq \left(\prod_{m=1}^{M}(1-\mu_m)\right)^{c \log T}\frac{M}{\mu_{\text{min}}} = \frac{M}{\mu_{\text{min}}T}\\
	&\hspace{0.18in} \text{and assign}\sum_{i=1}^{M c \log T}\prod_{j=1}^{i}(1-\mu_{k_m^*}(t)) = \phi.
	\end{align*}
	It follows that,
	\begin{align}
	\label{eq:age_in_the_good_case-2}
	&\E\left[ \sum_{\tau=0}^{\infty}\prod_{i=0}^{\tau} \left(1-\mu_{k(t-i)}\right)\bigg | E_t\right] \leq \phi +\frac{M}{\mu_{\text{min}}T}.
	\end{align}
	Moreover, since c$\mu_{k(t)}^m \geq \mu_{\min}/M$, for all $t$, 
	\begin{align}
	\label{eq:age_in_the_bad_case-2}
	&\E\left[ \sum_{\tau=0}^{\infty}\prod_{i=0}^{\tau} \left(1-\mu_{k(t-i)}\right)\bigg | E_t^{c}\right] \leq \frac{M}{\mu_{\text{min}}}.
	\end{align}	
	\noindent Note that
	\begin{align}
	E_{t}^c &= \bigcup_{\tau=t - c \log T + 1}^{t} \bigg\{A(\tau) \cup B(\tau) \bigg\} \nonumber \\
	\mathbbm{1}_{E_{t}^c } &\leq \sum_{\tau=t - c \log T + 1}^{t} \bigg( \mathbbm{1}_{A(\tau)} + \mathbbm{1}_{B(\tau)}  \bigg).
	\label{eq:expectation_indicator-2}
	\end{align}
	
	\noindent From \eqref{eq:doubleExpectation}, \eqref{eq:age_in_the_good_case-2}, \eqref{eq:age_in_the_bad_case-2}, and \eqref{eq:expectation_indicator-2}, and set  $\phi(T- M c \log T) = \phi^{'}$.\\
	Notice, constant $c$ scales with M, so define\\ 
	$c^{'} = \frac{-1}{\log (GM\left((1 - \mu_1), \cdots, (1 - \mu_M)\right) )} $, thus, $c^{'} = M c$
	\begin{align}
	\label{eq:usefulForQ}
	&\sum_{t = 1}^{T} \E[a_m^{\mathcal{P}}(t)]  \nn \\
	&= \sum_{t = 1}^{ M c \log T} \E[a_m^{\mathcal{P}}(t)] + \sum_{t =M c \log T + 1 }^{T} \E[a_{m}^{\mathcal{P}}(t)] \nn\\
	&\leq \frac{M c \log T}{\mu_{\text{min}}} + \phi^{'} +\frac{M (T- Mc \log T)}{\mu_{\text{min}}T}\nn \\
	&\hspace{.15in} +\frac{M}{\mu_{\text{min}}} \E \left[\sum_{t = M c \log T + 1}^{T} \mathbbm{1}_{E_{t}^c } \right]\\
	& \leq \frac{M (c^{'} \log T + 1)}{\mu_{\text{min}}} + \phi^{'}+ \frac{M}{\mu_{\text{min}}} \nn \\
	&\hspace{.15in}\times \left( \sum_{t = c^{'} \log T}^{T} \E \left[\sum_{\tau=t - c^{'} \log T + 1}^{t} \big( \mathbbm{1}_{A(\tau)} + \mathbbm{1}_{B(\tau)} \big)\right] \right) \nn\\
	& \leq \phi^{'} +\frac{M (c^{'}  \log T + 1)}{\mu_{\text{min}}} +\frac{M c^{'}  \log T}{\mu_{\text{min}}} \E \left[\sum_{t = 1}^{T}  \big( \mathbbm{1}_{A(t)} + \mathbbm{1}_{B(t)} \big) \right]\nn \\
	& \leq  \phi^{'}+ \frac{M}{\mu_{\text{min}}} + \frac{M (c^{'}  \log T)}{\mu_{\text{min}}} \left(1+ \E \left[ \sum_{t = 1}^{T}  \big( \mathbbm{1}_{A(t)} + \mathbbm{1}_{B(t)} \big) \right]\right)\nn\\
	& \leq \phi^{'} + \frac{T}{\mu^{*}} + \frac{M}{\mu_{\text{min}}} + \frac{M (c^{'} \log T)}{\mu_{\text{min}}} \bigg(1+ \E \bigg[ \sum_{t = 1}^{T}\big(\mathbbm{1}_{A(t)} \nn\\
	&\hspace{1.8in} + \sum_{i\neq m} \mathbbm{1}_{B_i(t)}\big) \bigg]\bigg)\nn\\
	& =  \phi^{'} + \frac{T}{\mu^{*}} + \frac{M}{\mu_{\text{min}}} + \frac{M (c^{'} \log T)}{\mu_{\text{min}}} \bigg(1+\E\big[N_1(\textbf{K}_m(T)) \nn\\
	&\hspace{1.8in} + N_2(\textbf{K}_m(T))\big]\bigg) \nn\\
	& =  \phi^{'} + \frac{M}{\mu_{\text{min}}} + \frac{M (c^{'}  \log T)}{\mu_{\text{min}}} (1+\E\left[N(\textbf{K}(T))\right]).
	\end{align}
Regret is defined as
\begin{align}
	\mathcal{R}^m_{\mathcal{P}}(t) =& \sum_{t = 1}^{T} \E[a_m^{\mathcal{P}}(t)]  - \sum_{t = 1}^{T} \E[a_m^*(t)] \nn \\
	\leq& \hspace{.2in}\sum_{t = 1}^{T} \E[a_m^{\mathcal{P}}(t)] -\sum_{t = 1 + M c \log T}^{T} \E[a_m^*(t)] \nn \\
	\leq& \hspace{.1in}\sum_{t = 1}^{T} \E[a_m^{\mathcal{P}}(t)] - \sum_{t = 1 + M c \log T}^{T} \phi \nn \\
	=& \sum_{t = 1}^{T} \E[a_m^{\mathcal{P}}(t)] - \phi^{'} \nn \\
	=& \frac{M}{\mu_{\text{min}}} + \frac{M (c^{'}  \log T)}{\mu_{\text{min}}} (1+\E\left[N(\textbf{K}(T))\right]).
	\end{align}

\end{proof}
\section{Conclusion}
We model a multi-source multi-channel setting using Age-of-Information bandits with decentralized users. We first characterize the oracle policy, namely, the round-robin policy, and demonstrate that the upper-bound of AoI-regret of the existing Distributed Learning with Fairness (DLF) policy scales as $\OO(M^2 N \log^2 T)$. However, proving the optimality of the oracle policy for the general setting is still an open problem. We then propose two other AoI-agnostic policies: Distributed Learning-based Thompson Sampling (DL-TS) and Distribute Learning-based Hybrid policy (DLH), with different trade-offs between exploration and the number of collisions. These policies only utilize the past arm-pulls in deciding the next arm. We also present AoI-aware policies that incorporate the current value of AoI into the arm selections. Through simulations, we show that the AoI-aware policies lead to fewer collisions and thus, outperform their agnostic counterparts.

\section*{Acknowledgment}
We thank Shourya Pandey of the Department of Computer Science and Engineering at Indian Institute of Technology Bombay for mathematical insights in the proofs.

\bibliographystyle{IEEEtran}
\bibliography{conference_101719}

\begin{thebibliography}{10}
\providecommand{\url}[1]{#1}
\csname url@samestyle\endcsname
\providecommand{\newblock}{\relax}
\providecommand{\bibinfo}[2]{#2}
\providecommand{\BIBentrySTDinterwordspacing}{\spaceskip=0pt\relax}
\providecommand{\BIBentryALTinterwordstretchfactor}{4}
\providecommand{\BIBentryALTinterwordspacing}{\spaceskip=\fontdimen2\font plus
\BIBentryALTinterwordstretchfactor\fontdimen3\font minus
  \fontdimen4\font\relax}
\providecommand{\BIBforeignlanguage}[2]{{%
\expandafter\ifx\csname l@#1\endcsname\relax
\typeout{** WARNING: IEEEtran.bst: No hyphenation pattern has been}%
\typeout{** loaded for the language `#1'. Using the pattern for}%
\typeout{** the default language instead.}%
\else
\language=\csname l@#1\endcsname
\fi
#2}}
\providecommand{\BIBdecl}{\relax}
\BIBdecl

\bibitem{kaul2011minimizing}
S.~Kaul, M.~Gruteser, V.~Rai, and J.~Kenney, ``Minimizing age of information in
  vehicular networks,'' in \emph{2011 8th Annual IEEE Communications Society
  Conference on Sensor, Mesh and Ad Hoc Communications and Networks}.\hskip 1em
  plus 0.5em minus 0.4em\relax IEEE, 2011, pp. 350--358.

\bibitem{kosta2017age}
A.~Kosta, N.~Pappas, and V.~Angelakis, ``Age of information: A new concept,
  metric, and tool,'' \emph{Foundations and Trends in Networking}, vol.~12,
  no.~3, pp. 162--259, 2017.

\bibitem{bhandari2020age}
K.~{Bhandari}, S.~{Fatale}, U.~{Narula}, S.~{Moharir}, and M.~K. {Hanawal},
  ``Age-of-information bandits,'' in \emph{2020 18th International Symposium on
  Modeling and Optimization in Mobile, Ad Hoc, and Wireless Networks (WiOPT)},
  2020, pp. 1--8.

\bibitem{OpportunisticSpectrum}
Y.~{Gai} and B.~{Krishnamachari}, ``{Distributed Stochastic Online Learning
  Policies for Opportunistic Spectrum Access},'' \emph{IEEE Transactions on
  Signal Processing}, vol.~62, no.~23, pp. 6184--6193, 2014.

\bibitem{thompson1933likelihood}
W.~R. Thompson, ``On the likelihood that one unknown probability exceeds
  another in view of the evidence of two samples,'' \emph{Biometrika}, vol.~25,
  no. 3/4, pp. 285--294, 1933.

\bibitem{sombabu2018age}
B.~Sombabu and S.~Moharir, ``Age-of-information aware scheduling for
  heterogeneous sources,'' in \emph{Proceedings of the 24th Annual
  International Conference on Mobile Computing and Networking}, 2018, pp.
  696--698.

\bibitem{tripathi2017age}
V.~Tripathi and S.~Moharir, ``Age of information in multi-source systems,'' in
  \emph{GLOBECOM 2017-2017 IEEE Global Communications Conference}.\hskip 1em
  plus 0.5em minus 0.4em\relax IEEE, 2017, pp. 1--6.

\bibitem{tripathi2019whittle}
V.~Tripathi and E.~Modiano, ``A whittle index approach to minimizing functions
  of age of information,'' in \emph{2019 57th Annual Allerton Conference on
  Communication, Control, and Computing (Allerton)}.\hskip 1em plus 0.5em minus
  0.4em\relax IEEE, 2019, pp. 1160--1167.

\bibitem{jhunjhunwala2018age}
P.~R. Jhunjhunwala and S.~Moharir, ``Age-of-information aware scheduling,'' in
  \emph{2018 International Conference on Signal Processing and Communications
  (SPCOM)}.\hskip 1em plus 0.5em minus 0.4em\relax IEEE, 2018, pp. 222--226.

\bibitem{kadota2018optimizing}
I.~Kadota, A.~Sinha, and E.~Modiano, ``Optimizing age of information in
  wireless networks with throughput constraints,'' in \emph{IEEE INFOCOM
  2018-IEEE Conference on Computer Communications}.\hskip 1em plus 0.5em minus
  0.4em\relax IEEE, 2018, pp. 1844--1852.

\bibitem{liu2010distributed}
K.~Liu and Q.~Zhao, ``Distributed learning in multi-armed bandit with multiple
  players,'' \emph{IEEE Transactions on Signal Processing}, vol.~58, no.~11,
  pp. 5667--5681, 2010.

\bibitem{besson2018multi}
L.~Besson and E.~Kaufmann, ``Multi-player bandits revisited,'' in
  \emph{Algorithmic Learning Theory}, 2018, pp. 56--92.

\bibitem{hanawal2018multi}
M.~K. Hanawal and S.~J. Darak, ``Multi-player bandits: A trekking approach,''
  \emph{arXiv preprint arXiv:1809.06040}, 2018.

\bibitem{boursier2019sic}
E.~Boursier and V.~Perchet, ``Sic-mmab: synchronisation involves communication
  in multiplayer multi-armed bandits,'' in \emph{Advances in Neural Information
  Processing Systems}, 2019, pp. 12\,071--12\,080.

\bibitem{kurose2010computer}
J.~Kurose and K.~Ross, ``Computer networks: A top down approach featuring the
  internet,'' \emph{Peorsoim Addison Wesley}, 2010.

\bibitem{auer2002finite}
P.~Auer, N.~Cesa-Bianchi, and P.~Fischer, ``Finite-time analysis of the
  multiarmed bandit problem,'' \emph{Machine learning}, pp. 235--256, 2002.

\bibitem{bullen2013handbook}
P.~S. Bullen, \emph{Handbook of means and their inequalities}.\hskip 1em plus
  0.5em minus 0.4em\relax Springer Science \& Business Media, 2013, vol. 560.

\bibitem{kaufmann2012thompson}
E.~Kaufmann, N.~Korda, and R.~Munos, ``{Thompson sampling: An asymptotically
  optimal finite-time analysis},'' in \emph{International conference on
  algorithmic learning theory}.\hskip 1em plus 0.5em minus 0.4em\relax
  Springer, 2012, pp. 199--213.

\end{thebibliography}

\end{document}